\documentclass[a4paper,11pt]{article}
\usepackage{amsmath, amssymb, graphicx, epstopdf, parskip, amsthm, fancyhdr, dsfont, algorithm, algpseudocode, color, hyperref}
\usepackage[round]{natbib}

\newtheorem{prop}{Proposition}

\theoremstyle{definition}

\newtheorem{rmk}{Remark}

\newcommand{\N}{\mathbb{N}}

\newcommand{\R}{\mathbb{R}}

\newcommand{\F}{\mathcal{F}}

\newcommand{\E}{\mathbb{E}}

\newcommand{\Prob}{\mathbb{P}}
\newcommand{\Q}{\mathbb{Q}}

\title{Inference and rare event simulation for stopped Markov processes via reverse-time sequential Monte Carlo}
\author{Jere Koskela \\
	\texttt{koskela@math.tu-berlin.de}\\
	\small Institut f\"ur Mathematik \\
	\small Technische Universit\"at Berlin \\
	\small Stra{\ss}e des 17.~Juni 136, 10623 Berlin, Germany
	\and
	Dario Span\`{o} \\
	\texttt{d.spano@warwick.ac.uk}\\
	\small Department of Statistics \\
	\small University of Warwick \\
	\small Coventry, CV4 7AL, UK
	\and
	Paul A. Jenkins \\
	\texttt{p.jenkins@warwick.ac.uk}\\
	\small Departments of Statistics \& Computer Science\\
	\small University of Warwick \\
	\small Coventry, CV4 7AL, UK
}
\date{\today}

\begin{document}

\maketitle

\begin{abstract}
We present a sequential Monte Carlo algorithm for Markov chain trajectories with proposals constructed in reverse time, which is advantageous when paths are conditioned to end in a rare set.
The reverse time proposal distribution is constructed by approximating the ratio of Green's functions in Nagasawa's formula.
Conditioning arguments can be used to interpret these ratios as low-dimensional conditional sampling distributions of some coordinates of the process given the others.
Hence the difficulty in designing SMC proposals in high dimension is greatly reduced.
We illustrate our method on estimating an overflow probability in a queueing model, the probability that a diffusion follows a narrowing corridor, and the initial location of an infection in an epidemic model on a network.
\end{abstract}

\section{Introduction}
\label{introduction}
Sequential Monte Carlo (SMC) is a very general technique for sampling from a sequence of complicated distributions of increasing dimension and known pointwise up to a normalising constant.
For a comprehensive introduction, see e.g.~\citep{Doucet01, Liu01, dou:joh:2011}. 
Briefly, a ``cloud'' of weighted particles is extended from one distribution to the next by a combination of sequential importance sampling and resampling. 
Each set of weighted particles then forms an empirical approximation of each subsequent distribution, provided adjacent distributions are sufficiently similar. 
This is typically the case when interest is in inference from a sequence of observations from e.g.~a hidden Markov model, and thus SMC finds widespread use in the filtering literature \citep{Doucet01, Liu01, DelMoral04, fea:2008, dou:joh:2011}.

In this paper we present a SMC algorithm for sampling trajectories of Markov chains started at the last exit time of an initial set $I$, and killed at the first hitting time of a terminal set $T$.
By last exit time, we mean that trajectories of interest are not allowed to re-enter $I$ before hitting $T$.
Such Markov trajectories have a wide range of applications, such as population genetics \citep{Griffiths94c, Stephens00, DeIorio04a, DeIorio04b, Birkner08, Hobolth08, Griffiths08, Birkner11, Koskela15}, mathematical finance \citep{Casella08}, neuroscience \citep{Bibbona13}, physics \citep{DelMoral05, Johansen06}, and engineering \citep{Blom07, Lezaud10}.
This problem is non-standard in SMC \citep{fea:2008} because the dimension of a particle can be random, and the intermediate target distribution of interest, namely that of a partially reconstructed chain conditioned on its eventual hitting of $T$ before returning to $I$, is usually unavailable. 
In practice the intermediate target may be replaced with its unconditioned counterpart, in which case the current importance weight of a particle is in very poor correlation with its final weight.
This renders the resampling step of a SMC algorithm counterproductive.

To circumvent this problem, our algorithm consists of proposing trajectories in \emph{reverse} time, started from $T$ and propagated until the first hitting time of $I$. 
The distribution of the reverse-time process is typically also intractable, but SMC techniques are used to mitigate the aforementioned problems, and produce properly weighted particle ensembles for unbiased estimators.

In practice, the utility of our method is in approximating quantities of the form $\E_{ \mu }[ f( X ) ]$, where $f$ is an integrable function, $X$ a Markov trajectory connecting $I$ to $T$ without returning to $I$, and $\mu$ is an initial distribution placing full mass on $I$.
Dependency of $f$ on the hitting time of $T$ can also be incorporated without difficulty, as will be seen in later sections.
The method is fairly general, but we can expect it to be particularly efficient in contexts where
\begin{itemize}
\item [(i)] the function $f$ depends only on the terminal point in $T$, or a small length of trajectory preceding it,
\item [(ii)] the support of $f|_T$ is small in terms of dimension and/or volume, while $I$ is high dimensional and/or has large volume,
\item [(iii)] the majority of the contribution to $\E_{ \mu }[ f( X ) ]$ arises from a region of low conditional probability given that the chain has hit $T$, and
\item [(iv)] the process of interest is high-dimensional and transitions only alter a small number of components at a time.
\end{itemize}

Properties (i) and (ii) ensure that time-reversal is an effective strategy.
In the extreme case where $f$ is the indicator function of a singleton in $T$ (corresponding to estimation of a conditional hitting density), $I$ is a set which is hit by the reverse-time dynamics in finite time with probability 1, and $T$ is a reverse-time entrance boundary, the optimal proposal distribution leading to zero variance estimators is the unconditional time-reversal.
These conditions are very restrictive, but the fact that all of the examples in Section \ref{examples} violate at least one of them demonstrates that reverse-time proposals can still yield efficient algorithms under the milder conditions (i)-(iv).

Property (iii) is helpful in ensuring that $T$ acts like a reverse-time entrance boundary with high probability, as proposal trajectories will naturally be pushed away from the rare hitting point of $T$ and back towards $I$.
Property (iv) means that it is only necessary to come up with a proposal distribution for the coordinates which differ between transitions, given the value of all other coordinates.
This dimensionality reduction can greatly reduce the difficulty of designing proposal distributions in high dimension.
Proposition \ref{green_prop} in Section \ref{framework} provides a precise formulation, and Sections \ref{atm} and \ref{sis} contain concrete examples.

Our method is reminiscent of existing forward-in-time SMC strategies known as multilevel SMC \citep{DelMoral04, Chen05, Jasra14}, against which our method should be considered as a complementary rather than an alternative approach.
Indeed, forward-in-time multilevel methods are well-suited for contexts involving the opposite to the point (ii) above, i.e.~an initial set $I$ of small volume and/or dimension, and a function $f$ with support in $T$ of large volume and/or dimension.

While time-reversal has proved to be a successful tool for inference in population genetics \citep{Griffiths94c, Stephens00, DeIorio04a, DeIorio04b, Birkner08, Griffiths08, Hobolth08, Birkner11, Koskela15}, and has also been used in rare event simulation \citep{Frater89, Anantharam90, Frater90, Shwartz93} and physics \citep{Jarzynski06}, its use in combination with SMC has been limited to a few examples in population genetics \citep{Chen05, Jenkins12, Koskela15}. 
It is in our opinion somewhat surprising that such an approach has not been combined with SMC for more general inference. 
\citet{Lin10} use an auxiliary distribution operating in reverse time within SMC, but only to improve resampling, not to construct trajectories.
Their time-reversal is also not constructed as an approximation of the reverse-time dynamics of their model.
\citet{Jasra14} use a coalescent-based example from population genetics to motivate their work, and highlight it as an example of more general SMC in reverse time, but still formulate their results forwards in time.

The examples presented in this paper are diverse: we consider an overflow probability in a queueing system, the probability that a diffusion trajectory remains contained in a narrow strip, and the initial location of an infection in an epidemic model on a network.
These examples demonstrate the utility of reverse time proposals beyond the coalescent setting.
The third example in particular would be challenging to solve with a forwards-in-time algorithm, as it is difficult to design a function whose level sets connect the initial condition to the desired rare event in a natural way (see \eqref{reaction_coord} in the next section).
Such a function is crucial for the implementation of forwards-in-time multilevel SMC algorithms, but not needed in our reverse time method.

The rest of the paper is laid out as follows.
In Section \ref{preliminaries} we formulate the rare event problem of interest and review forwards-in-time SMC.
In Section \ref{framework} we present our reverse-time algorithm, and show how it naturally yields a dimension reduction in the SMC design task.
Section \ref{examples} contains our example simulations, and Section \ref{discussion} concludes with a discussion.

\section{Problem formulation and SMC}\label{preliminaries}

Consider the canonical probabilty space
\begin{equation}\label{mc}
\left( \Omega := \prod_{n=0}^{\infty} E_n, \ \F := \bigotimes_{n=0}^{\infty} \F_n , \ \{ X_n \}_{ n = 0 }^{ \infty }, \ \Prob_{\mu} \right),
\end{equation}
where $\Prob_{\mu}$ is defined via its finite dimensional distributions as
\begin{equation}\label{fw_prob}
\Prob_{\mu}( X_{0:n} \in dx_{0:n} ) = \mu(dx_0) \prod_{ i = 0 }^{ n - 1 } P( x_i, dx_{ i + 1 } ).
\end{equation}
Here $x_{0:n} := ( x_0, \ldots, x_n )$, and $P$ is a given transition kernel.
We assume both $P$ and $\mu$ can be evaluated pointwise, but do not assume \eqref{mc} is stationary or even has a stationary distribution.
Following the convention of \cite{DelMoral04} (see page 51 in particular), we will also refer to \eqref{mc} as the canonical Markov chain.

As in the previous section, let $I \subset \N \times \Omega$ be an initial set, and $T \subset \N \times \Omega$ be a space-time target set, which we assume has finite expected hitting time under the dynamics of $\{ n, X_n \}_{ n = 0 }^{ \infty }$.
We assume that $\mu( I ) = 1$, and are interested in approximating functionals of trajectories
\begin{equation}\label{functional}
\E_{ \mu }[ f( \tau_T, X_{ 0 : \tau_T } ) | \tau_T < \tau_I ]
\end{equation}
for integrable functions $f$, where for a set $A \in \mathcal{ B }( \N ) \times \mathcal{ B }( \Omega )$,
\begin{equation*}
\tau_A := \inf\{ n > 0 : ( n, X_n ) \in A \}
\end{equation*}
denotes the hitting time of $A$, and $\E_{ \mu }$ denotes expectation with respect to $\Prob_{ \mu }$.
We emphasize that these trajectories are defined between the \emph{last exit time} of $I$ and the \emph{hitting time} of $T$.
In particular, in our problem formulation trajectories of interest cannot re-enter $I$ before hitting $T$.
As an example, take $T$ to depend only on space, and consider the hitting probability (resp.~density) of a point $x \in T$ whenever $\Omega$ is discrete (resp.~continuous), before any other point of $T$.
The corresponding functional is
\begin{equation*}
\E_{ \mu }[ f( \tau_T, X_{ 0 : \tau_T } ) | \tau_T < \tau_I ] = \E_{ \mu }[ \mathds{ 1 }_{ \{ x \} }( X_{ \tau_T } ) | \tau_T < \tau_I ].
\end{equation*}
As per point (iii) in Section \ref{introduction}, we assume that $f$ places most of its mass on trajectories with terminal states that are unlikely under the dynamics of \eqref{mc}.
Thus, approximating \eqref{functional} can be seen as a rare event simulation problem of sampling trajectories with unlikely terminal states.

Sequential Monte Carlo is a standard method in rare event simulation \citep{Rubino09}, and consists of constructing multiple trajectories in parallel by alternating between sequential importance sampling and resampling steps.
Sequential importance sampling consists of sampling from a sequence of proposal distributions to build up a single, high-dimensional realisation.
The proposals are typically not the conditional distributions of the model of interest, and so samples must be reweighted by the Radon-Nikodym derivative of the model and the proposal:
\begin{align*}
\E_{ \mu }[ f( X_{ 0 : \tau_T } ) ] &= \int f( x_{ 0 : \tau_T } ) \Prob_{\mu}( dx_{0 : \tau_T } ) = \int f( x_{ 0 : \tau_T } ) \frac{ d\Prob_{\mu} }{ d \Q_{ \eta } }( x_{ 0 : \tau_T } ) \Q_{ \eta }( dx_{0 : \tau_T } ) \\
&= \int f( x_{ 0 : \tau_T } ) \frac{ \mu( x_0 ) }{ \eta( x_0 ) } \prod_{ n = 0 }^{ \tau_T - 1} \frac{ P( x_n, x_{ n + 1 } ) }{ Q( x_n, x_{ n + 1 } ) } Q( x_n, dx_{ n + 1 } ).
\end{align*}
Here $\Q_{ \eta }$ is a proposal distribution with initial law $\eta$ and one-step transition kernel $Q$, and with $\Prob_{ \mu } \ll \Q_{ \eta }$.
The resampling step is applied at intermediate times $t < \tau_T$, and consists of duplicating promising particles with high weight $\frac{ \mu( x_0 ) }{ \eta( x_0 ) } \prod_{ n = 0 }^t \frac{ P( x_n, x_{ n + 1 } ) }{ Q( x_n, x_{ n + 1 } ) }$, while discarding particles with low weight.

Good choices of $\Q_{ \eta }$ and resampling schedule can dramatically reduce the variance of estimators, achieving the same asymptotic efficiency as the popular multilevel splitting method of rare event simulation under mild conditions \cite{Cerou11}.
On the other hand, poor choices of $\Q_{ \eta }$ can yield estimators with higher variance than naive Monte Carlo \citep{Glasserman97}.
The optimal proposal is the law of the process conditioned on the event of interest, but computing this law involves the quantity of interest so the optimal algorithm is unimplementable in practice.
The typical approach is to approximate the optimal proposal distribution using large deviations \citep{Sadowsky90}.

Without resampling, the variance of estimators typically increases exponentially in the number of sequential steps even with well designed proposal distributions \citep{Doucet01, Liu01}.
In the context of stopped processes, resampling should take into account the weight of the particle \emph{and} the progress it has made towards the target set.
This is achieved by introducing a sequence of intermediate sets, propagating all samples until they hit the next set, and perform resampling based on current weights once all particles have been stopped.
This has been alternatively termed multilevel SMC or stopping time resampling in \citep[Section 12.2]{DelMoral04} and \citep{Chen05}, respectively.
The good performance of our algorithm also relies on resampling, and hence should be regarded as reverse-time multilevel SMC.

Like the proposal distribution, the choice of resampling schedule also has a strong impact on the efficiency of the SMC algorithm.
Few theoretical guidelines are available, though developments have been made in determining good schedules adaptively both for multilevel splitting \citep{Cerou12} and SMC \citep{Jasra14}.
As with all multilevel splitting algorithms, the method of \citep{Cerou12} requires a reaction coordinate: a tractable function 
\begin{equation}\label{reaction_coord}
\Psi : \Omega \mapsto \R,
\end{equation}
such that both the rare event being estimated, as well as intermediate resampling stages, can be expressed as nested level sets of $\Psi$.
The interested reader is directed to e.g.~Section 2 of \citep{Cerou12} for details.
Our method has the advantage of not requiring such a function.
The examples in Section \ref{examples} will show that effective reaction coordinates are not always available, meaning that the advantage of not requiring one can be substantial.
The particle MCMC method of \citet{Jasra14} for designing efficient resampling schedules adaptively applies in our setting.

In summary, the main result of this paper is a general way to specify a SMC algorithm with a proposal distribution that constructs trajectories in reverse time.
We will show how such proposals can be specified in very general settings, and that they yield good performance when combined with stopping time resampling.
While performing stopping time resampling is crucial for the performance of our algorithm, we do not investigate optimisation of the resampling schedule.
Our results apply to settings in which the conditions (i)-(iv) in Section 1 hold, and all three of the initial law $\mu$, the forward transition density $P$ and the reverse time proposal $Q$ can be evaluated pointwise up to normalising constants.

\section{Time-reversal as a SMC proposal distribution}\label{framework}

In this section we review some relevant facts about time-reversal of Markov chains and introduce our generic SMC proposal distribution.
Concrete examples can be found in Section \ref{examples}.

We define the time-reversal of \eqref{mc} by extending the chain to the negative time-axis, and letting
\begin{equation*}
\left( \prod_{ n = 0 }^{ - \infty } E_n, \ \bigotimes_{ n = 0}^{ - \infty } \F_n ,\  \left\{ \widetilde{X}_n \right\}_{ n = 0 }^{ - \infty },  \ \widetilde{\Prob}_{\nu} \right)
\end{equation*}
denote the reverse-time chain.
Note that the initial time-indices are set to 0 by convention, and are not necessarily intended to coincide with the starting time of \eqref{mc}.
The law $\widetilde{\Prob}_{\nu}$ is again defined via its finite dimensional distributions as
\begin{equation*}
\widetilde{\Prob}_{\nu}( dx_{ 0 : - n } ) = \nu(dx_0) \prod_{ i = 0 }^{ - n + 1 } \widetilde{P}( x_i, dx_{ i - 1 } ).
\end{equation*}
Reverse time proposal distributions akin to \eqref{nagasawa} have been studied previously by \citet{Birkner11} for certain population genetic models.
The reverse transition kernel is related to its forward counterparts via Nagasawa's formula \citep[Section VI.46]{Rogers94}:
\begin{equation}\label{nagasawa}
\widetilde{P}( x_i, x_{ i - 1 } ) = \frac{ G( \mu,  x_{ i - 1 } ) }{ G( \mu, x_i ) } P( x_{ i - 1 }, x_i ),
\end{equation}
where for a measurable set $A$,
\begin{equation*}
G( \mu, A ) := \E_\mu\left[ \sum_{ i = 0 }^{ \tau_T } \mathds{ 1 }_A( X_i ) \right]
\end{equation*}
is the Green's function of \eqref{mc}, and $\E_{\mu}$ denotes expectation with respect to $\Prob_{\mu}$.
When $A = \{ z \}$ is a null set and the state space is continuous, we interpret the Green's function $G( \mu, z )$ as a density, which is assumed to exist.
For simplicity we assume that all the transition kernels (resp. Green's functions) are absolutely continuous with respect to the same reference measure (e.g.~Lebesgue or counting measure), and we will use the same notation for both transition kernels (resp.~Green's functions) and their densities. 

The Green's functions in \eqref{nagasawa} cannot be computed in most cases of interest, but their qualitative behaviour can often be described.
We assume that such a description is available, and that it is possible to write down an approximating family of tractable functions with similar qualitative behaviour.
It is not necessary for the match to be very precise, though better approximations yield more efficient SMC algorithms.

Our strategy for defining an SMC proposal is as follows:
\begin{enumerate}
\item Design an approximate Green's function $\hat{G}( \mu, x )$ to be substituted into \eqref{nagasawa} to yield an approximate reverse-time transition kernel $\hat{ P }$ and a proposal Markov chain
\begin{equation}\label{proposal}
\left( \prod_{ n = 0 }^{ - \infty } E_n, \prod_{ n = 0}^{ - \infty } \F_n , \left\{ \hat{X}_n \right\}_{ n = 0 }^{ - \infty }, \hat{\Prob}_{\nu} \right)
\end{equation}
where $\hat{\Prob}_{\nu}$ is defined from its finite dimensional distributions via $\hat{P}$ as before.
\item\label{start_cond} If necessary, modify $\hat{\Prob}_{\nu}$ locally to incorporate first hitting time constraints by preventing \eqref{proposal} from returning to $T$ upon leaving it. 
\item\label{end_cond} If necessary, further modify $\hat{\Prob}_{\nu}$ locally to ensure the expected hitting time of $I$ by \eqref{proposal} is finite to ensure finite expected run time.
\end{enumerate}
These steps can seem abstract because of their generality, but we will see in Section \ref{examples} that they can be carried out in many cases of interest.

Steps \ref{start_cond} and \ref{end_cond} could be incorporated automatically and more efficiently by considering the time-reversal of \eqref{mc} after an appropriate Doob's $h$-transform \citep[Section VI.45]{Rogers94}, that is by considering forwards in time transition densities conditioned on hitting the rare event.
However, the $h$-transform is typically intractable, whereas local modifications are widely implementable and still result in efficient algorithms when the dominant contribution to $\E_{ \mu }[ f( X_{ 0 : \tau_T } ) ]$ arises from a region of low $\Prob_{ \mu }$-probability and $I$ lies in a region of high $\Prob_{ \mu }$-probability, as then the ratio of Green's functions will drive the process away from $T$ and towards $I$ without any conditioning (c.f.~Property (iii) in Section \ref{introduction}).

Proposition \ref{green_prop} presents a practical way of designing approximate ratios of Green's functions for a wide class of models.
For notational simplicity we assume a countable state space, but the same argument holds for continuous state spaces provided the Green's densities $g( x, z )$ exist.
\vskip 11pt
\begin{prop}\label{green_prop}
Consider a transition of the Markov chain \eqref{mc} from $x_{ n - 1 }$ to $x_n$, and suppose the state space can be partitioned so that $x_{ n - 1 } = ( z, y )$ and $x_n = ( z, \bar{ y } )$.
Assume that the conditional sampling distribution (CSD) 
\begin{equation*}
\pi( y | z ) := \Prob_{ \mu }( Y_n = y | Z_n = z )
\end{equation*}
is independent of $n \in \N$ for $\Prob_\mu$-almost every $z$.
Then the ratio of Green's functions in \eqref{nagasawa} cancels to the ratio of CSDs:
\begin{equation*}
\frac{ G( \mu, ( z, y ) ) }{ G( \mu, ( z, \bar{ y } ) ) } = \frac{ \pi( y | z ) }{ \pi( \bar{ y } | z ) }.
\end{equation*}
\end{prop}
\begin{proof}
Let $\partial$ be a cemetery state, and define the process
\begin{equation*}
X_t^{ \partial } = \begin{cases}
X_t &\text{ if } t \leq \tau_T \\
\partial &\text{ otherwise}
\end{cases}.
\end{equation*}
Note that the laws of $\{ X_n^{ \partial } \}_{ n = 0 }^{ \tau_T }$ and $\{ X_n \}_{ n = 0 }^{ \tau_T }$ coincide, and thus so do their Green's functions evaluated at states $( z, y ) \in ( T \cup \partial)^c$.
Hence, for any such state, Fubini's theorem and conditioning on $Z_n = z$ yield
\begin{align*}
G( \mu, ( z,  y ) ) &= \E_{ \mu }\left[ \sum_{ n = 0 }^{ \infty } \mathds{ 1 }_{ \{ z, y \} }( Z_n^{ \partial }, Y_n^{ \partial } ) \right] = \sum_{ n = 0 }^{ \infty } \E_{ \mu }\left[ \mathds{ 1 }_{ \{ z, y \} }( Z_n^{ \partial }, Y_n^{ \partial } ) \right] \\
&= \sum_{ n = 0 }^{ \infty } \E_{ \mu }\left[ \E_{ \mu }\left[ \mathds{ 1 }_{ \{ y \} }( Y_n^{ \partial } ) | Z_n^{ \partial } = z \right] \mathds{ 1 }_{ \{ z \} }( Z_i^{ \partial } ) \right],
\end{align*}
where $( Z_n^{ \partial }, Y_n^{ \partial } ) := X_n^{ \partial }$.
Now 
\begin{equation*}
\pi( y | z ) = \E_{ \mu }\left[ \mathds{ 1 }_{ \{ y \} }( Y_n^{ \partial } ) | Z_n^{ \partial } = z \right]
\end{equation*}
is independent of $n$ by assumption.
Thus
\begin{equation*}
G( \mu, ( z,  y ) ) = \pi( y | z ) \sum_{ n = 0 }^{ \infty } \E_{ \mu }\left[ \mathds{ 1 }_{ \{ z \} }( Z_n ) \mathds{ 1 }_{ n : \infty }( \tau_T ) \right] = \pi( y | z ) \sum_{ n = 0 }^{ \infty } \sum_{ t = 0 }^{ \infty } \Prob_{ \mu }( Z_n = z, \tau_T = t ).
\end{equation*}
Now note that the final double sum cancels whenever the Green's functions are evaluated as ratios, which completes the proof.
\end{proof}
\vspace{11pt}
\begin{rmk}\label{cancellation}
The hypothesis of Proposition \ref{green_prop} is a weak conditional stationarity condition, and is relatively mild.
However, since in practice we are rely on ad hoc approximations to ratio of Green's functions, it is possible to extend the scope of the reverse-time framework by defining the proposal distribution based on a family of approximate CSDs $\hat{ \pi }( y | z )$ even when Proposition \ref{green_prop} fails.
In such a case, the ratio of approximate CSDs does not approximate a ratio of Green's functions, but can still yield an efficient algorithm if the qualitative behaviour of the approximate CSD resembles that of the true Green's function.
For example, if $\pi(y | z )$ is not independent of time, but the time dependence in $\pi( y | z )$ is weak, then a CSD corresponding to a time homogeneous system may still yield a good approximation.

Because of their lower dimension, approximate CSDs are often much easier to design than either proposal kernels $\{ Q( \cdot, \cdot ) \}$ or approximate Green's functions $\{ \hat{ G }( \mu, \cdot ) \}$.
Indeed, we consider this dimensionality reduction in the design task one of the main advantages of the reverse-time framework.
\end{rmk}

Choosing a family of approximate CSDs $\{ \hat{ \pi }( \cdot | \cdot ) \}$ and applying Proposition \ref{green_prop} to \eqref{nagasawa} yields proposal transition probabilities of the form
\begin{equation*}
\hat{ P }( ( z, \bar{ y } ), ( z, y ) ) = \frac{ \hat{ \pi }( y | z ) }{ C( ( z, \bar{ y } ) ) \hat{ \pi }( \bar{ y } | z ) } P( ( z, y ), ( z, \bar{ y } ) ),
\end{equation*}
where $C( ( z, \bar{y} ) )$ is a normalising constant, and the notation for the state variables is as in Proposition \ref{green_prop}.
The corresponding incremental importance weight at step $n$ is 
\begin{equation*}
\frac{ C( ( z, \bar{ y } ) ) \hat{ \pi }( \bar{ y } | z ) }{ \hat{ \pi }( y | z ) }.
\end{equation*}
Once a proposal chain has been constructed, functionals of interest can be unbiasedly estimated as
\begin{align}
\E_{ \mu }[ f( X_{ 0 : \tau_T } ) ] &\approx \frac{ 1 }{ N } \sum_{ j = 1}^N f( x_{ 0: \tau_T^{ ( j ) } }^{ ( j ) } ) \frac{ d\Prob_{ \mu } }{ d\hat{\Prob}_{ \nu } }( x_{ 0 : \tau_T^{ ( j ) } }^{ ( j ) } ) \nonumber \\
&= \frac{ 1 }{ N } \sum_{ j = 1 }^N f( x_{ 0 : \tau_T^{ ( j ) } } ) \frac{ \mu( x_0^{ ( j ) } ) }{ \nu( x_{ \tau_T } ) } \prod_{ n = 1 }^{ \tau_T } \frac{ \hat{ \pi }( n( x_{ n - 1 }^{ ( j  ) }, x_n^{ ( j ) } ) | e( x_n^{ ( j ) }, x_{ n - 1 }^{ ( j ) } ) ) }{ \hat{ \pi }( n( x_n^{ ( j  ) }, x_{ n - 1 }^{ ( j ) } ) | e( x_n^{ ( j ) }, x_{ n - 1 }^{ ( j ) } ) ) } C( x_n^{ ( j ) } ) \label{estimator},
\end{align}
where $\{ x_{ 0 : \tau_T^{ ( j ) } }^{ ( j ) } \}_{ j = 1 }^N$ is a sample from the SMC algorithm that uses \eqref{proposal} as its proposal mechanism, $e( x, y )$ is the vector consisting of those entries of $x$ and $y$ which are equal, and $n( x, y )$ is the vector consisting of the entries of $x$ which are not equal to the corresponding entry of $y$.
The purpose of this somewhat bulky notation is to encode the splitting of the state space into coordinates which remain constant between transitions, and others which change, that is introduced in Proposition \ref{green_prop}.
A formal algorithm specification for sampling properly weighted particle ensembles is given below in Algorithm \ref{alg}.
\begin{algorithm}[!ht]
\caption{Reverse-time multilevel SMC}
\label{alg}
\begin{algorithmic}[1]
\Require Particle number $N$, approximate CSD $\hat{ \pi }$, stopping times $\{ \tau_i \}_{ i = n }^1$ such that $0 < \tau_n \leq \tau_{ n - 1 } \leq \ldots \leq \tau_1 = \tau_I$.
\For{j = 1 to N}
	\State $X_0^j \sim \nu$. \Comment{Initial particle locations.}
	\State $\bar{ X }_0^j \gets X_0^j$. \Comment{Duplicates used in resampling.}
	\State $w_j \gets \nu( X_0^j )^{ -1 }$. \Comment{Initial importance weights.}
	\State $k_j \gets 0, \bar{ k }_j \gets 0$. \Comment{Time indices.}
	\State $A_j \gets j$. \Comment{Ancestor indices.}
\EndFor
\For{i = n to 1} \Comment{Do for each level:}
	\If{$\operatorname{ESS}( w_{ 1 : N } ) < N / 2$} \Comment{Resampling check.}
		\For{j = 1 to N}
			\State $A_j \sim \sum_{ k = 1 }^N w_k \delta_k$. 
			\State $\bar{ k }_j \gets k_{ A_j }$.
			\State $\bar{ X }_{ k_j }^j \gets X_{ k_j }^j$.
		\EndFor
		\State Set $\bar{ w } \gets N^{ -1 } \sum_{ k = 1 }^N w_k$.
	\Else
		\For{j = 1 to N}
			\State $A_j \gets j$. 
			\State $\bar{ k }_j \gets k_j$.
			\State $\bar{ X }_{ k_j }^j \gets X_{ k_j }^j$.
		\EndFor
	\EndIf
	\For{j = 1 to N}
		\If{$\bar{ k }_j < \tau_i$} \Comment{First step in level.}
			\State $k_j \gets \bar{ k }_j + 1$.
			\State $X_{ k_j }^j \sim \frac{ \hat{ \pi }( n( \cdot, \bar{ X }_{ \bar{ k }_j }^{ A_j } ) | e( \cdot, \bar{ X }_{ \bar{ k }_j }^{ A_j } ) ) P( \cdot , \bar{ X }_{ \bar{ k }_j }^{ A_j } ) \mathds{ 1 }_{ T^c }( \cdot ) }{ \hat{ \pi }( n( \bar{ X }_{ \bar{ k }_j }^{ A_j }, \cdot ) | e( \cdot, \bar{ X }_{ \bar{ k }_j }^{ A_j } ) ) C( \bar{ X }_{ \bar{ k }_j }^{ A_j } ) }$. 
			\State $w_j \gets \bar{ w } \frac{ \hat{ \pi }( n( \bar{ X }_{ \bar{ k }_j }^{ A_j }, X_{ k_j }^j ) | e( X_{ k_j }^j, \bar{ X }_{ \bar{ k }_j }^{ A_j } ) ) C( \bar{ X }_{ \bar{ k }_j }^{ A_j } ) }{ \hat{ \pi }( n( X_{ k_j }^j,  \bar{ X }_{ \bar{ k }_j }^{ A_j } ) | e( X_{ k_j }^j, \bar{ X }_{ \bar{ k }_j }^{ A_j } ) ) }$.
			\State $A_j \gets j$.
		\EndIf
		\While{$k_j < \tau_i$} \Comment{Propagate until next level.}
			\State $X_{ k_j }^j \sim \frac{ \hat{ \pi }( n( \cdot, X_{ k_j }^j ) | e( \cdot, X_{ k_j }^j ) ) P( \cdot , X_{ k_j }^j ) \mathds{ 1 }_{ T^c }( \cdot ) }{ \hat{ \pi }( n( X_{ k_j }^j, \cdot ) | e( \cdot, X_{ k_j }^j ) ) C( X_{ k_j }^j ) }$.
			\State $w_j \gets w_j \frac{ \hat{ \pi }( n( X_{ k_j }^j,  X_{ k_j + 1 }^j ) | e( X_{ k_j + 1 }^j, X_{ k_j }^j ) ) C( X_{ k_j }^j ) }{ \hat{ \pi }( n( X_{ k_j + 1 }^j, X_{ k_j }^j ) | e( X_{ k_j + 1 }^j, X_{ k_j }^j ) ) }$.
			\State $k_j \gets k_j + 1$.
		\EndWhile
	\EndFor
\EndFor
\For{j = 1 to N}
	\State Set $w_j \gets w_j \mu\left( X_{ k_j }^j \right)$. \Comment{Account for entrance law $\mu$.}
\EndFor
\end{algorithmic}
\end{algorithm}
\vskip 11pt
\begin{rmk}\label{rev_time_strengths}
Approximating \eqref{estimator} can be computationally daunting if $f|_T$ takes non-negligible values in a high-dimensional or large (in terms of Lebesgue-volume) subset of $T$, which is the rationale for property (ii) in Section \ref{introduction}.
In such cases we do not expect the reverse-time approach to always be competitive with forwards-in-time algorithms, particularly if the initial set $I$ is also small and hence difficult for the reverse-time chain to hit.
\end{rmk}
\vskip 11pt
\begin{rmk}
Normalising constants $C(x)$ in \eqref{estimator} would all be identically equal to one if an algorithm using the true ratio of Green's functions could be implemented.
Thus, the realised values of these constants for a given approximation could be used to design proposal distributions adaptively from trial runs, at least for discrete systems where the constants can be computed.
We do not explore this possibility further in this paper.
\end{rmk}

We conclude this section with a discussion on the intuition of why a SMC proposal distribution of the form \eqref{proposal} should be efficient in the rare event simulation context outlined in Section \ref{preliminaries}.
Loosely speaking, efficient rare event simulators for dynamical systems succeed by closely mimicking the conditional dynamics given the event of interest. 
In general, it is very difficult to identify good approximations to these dynamics, as they vary markedly from the unconditional dynamics of the system.
In contrast, in reverse time the unconditional dynamics and the dynamics conditioned on a rare end point are very similar: both will rapidly leave the rare state, and return to regions of high probability.
This behaviour is captured by the ratio of Green's functions in \eqref{nagasawa}, from which it is easy to see that steps that are proposed most often are those which, forwards in time, are both compatible with the dynamics of the model ($P( x_{ i - 1 }, x_i ) \gg 0$), and move the state towards the desired region of low probability ($G( \mu, x_{ i - 1 } ) > G( \mu, x_i )$).
The reverse-time view permits initialising samples in the rare state of interest, and hence enables sampling trajectories that approximate the conditional dynamics given the rare event using only unconditional quantities --- namely $P$ and an approximation of $G( \mu, \cdot )$.

\section{Numerical examples}\label{examples}

In this section we present three simulation studies illustrating the reverse-time approach: the asynchronous transfer mode (ATM) network (Section \ref{atm}), the hyperbolic diffusion (Section \ref{hyper_diff}) and the susceptible-infected-susceptible (SIS) epidemic on a network (Section \ref{sis}).
All three examples fulfil conditions (i)-(iv) in Section \ref{introduction}, and consequently our reverse-time SMC outperforms a state-of-the-art adaptive multilevel splitting algorithm \citep{Cerou12} in the first two examples.
Designing an efficient forwards-in-time SMC or splitting algorithm to tackle the third example would be a formidable task, and hence beyond the scope of this paper.
In contrast, we show that designing an efficient algorithm in the reverse-time framework is straightforward.

\subsection{The asynchronous transfer mode network}\label{atm}

Our first example is the ATM network studied by \citet{Glasserman99} in the context of rare events.
The network consists of $K$ sources, each of which is either on or off.
Sources which are off do nothing, while sources which are on produce packets at rate $\lambda$.
Packets are serviced by a common server with rate $\mu$ using the first-in-first-out policy.
Off sources turn on at rate $\alpha_0$ and on sources turn off at rate $\alpha_1$.
The state of the system is specified as $( i, j ) \in \N_0 \times \{ 0, \ldots, K \},$ where $i$ denotes the number of packets in the queue and $j$ the number of on sources.

\citet{Glasserman99} estimated the probability of the queue length hitting a barrier $b \in \N$ before emptying, given an empty initial queue and $K \alpha_0 / ( \alpha_0 + \alpha_1 )$ on sources.
Reverse-time SMC could be used for this example by summing over all possible numbers of terminal on sources, but this results in a $K$-fold increase in computational burden and hence cannot be expected to be competitive with a forwards-in-time approach.
We focus instead on the joint probability of an initially empty queue hitting a barrier $b$ before emptying, with exactly $k$ sources on at the hitting time, and assume the initial number of on sources is $\operatorname{Bin}( K, \alpha_0 / ( \alpha_0 + \alpha_1 ) )$ distributed.
In this scenario a forwards-in-time algorithm would face the same difficulties as a reverse-time algorithm does in the scenario of \citet{Glasserman99}.

The initial set is $I = \cup_{ r = 0 }^K \{ ( 0, r ) \}$, the target set is $T = \cup_{ r = 0 }^K \{ ( 0, r ) \cup ( b, r ) \}$, the initial law is
\begin{equation*}
\mu( \{ ( 0, j ) \} ) = \binom{ K }{ j } \left( \frac{ \alpha_0 }{ \alpha_0 + \alpha_1 } \right)^j \left( \frac{ \alpha_1 }{ \alpha_0 + \alpha_1 } \right)^{ K - j },
\end{equation*}
and the quantity of interest is the hitting probability 
\begin{equation*}
\E_{ \mu }[ f( \tau_T, X_{ 0 : \tau_T } ) ] = \E_{ \mu }[ \mathds{ 1 }_{ \{ ( b, k ) \} }( X_{ \tau_T } ) ].
\end{equation*}
The proposal is specified by defining approximate conditional distributions of $i$ given $j$ and $j$ given $i$.
These are denoted by $\hat{ \pi }_i( i | j )$ and $\hat{ \pi }_j( j | i )$ respectively, and we choose them to be
\begin{align*}
\hat{ \pi }_i( i | j ) &\propto \left( \frac{ \lambda \max\{ j, 1 \} }{ \mu } \right)^i \text{ for } i \in 0:b \text{ and } j \in 0:K, \\
\hat{ \pi }_j( j | i ) &\propto \binom{ K }{ j } \left( \frac{ \alpha_0 }{ \alpha_0 + \alpha_1 } \right)^j \left( \frac{ \alpha_1 }{ \alpha_0 + \alpha_1 } \right)^{ K - j } \hat{ \pi }_i( i | j ) \text{ for } i \in 0 : b \text{ and } j \in 0 : K.
\end{align*}
The former is the true distribution of a queue with arrival rate $\lambda j$ and service rate $\mu$ whenever $\mu > \lambda j$, and well-defined otherwise as well because the range of possible values of $i$ is finite.
The minimum in $\hat{ \pi }_i( i | j )$ is to allow for positive queue lengths even when no sources are active.
We also employ stopping time resampling, with simulations stopped every time a new minimum queue length is reached in reverse time.
Resampling takes place if the effective sample size is below half of the number of particles once all simulations have been stopped.

To present a comparison of our method, we also implemented the adaptive multilevel splitting algorithm of \citet{Cerou12}.
This algorithm requires the user to specify a reaction coordinate $\Psi$ (see \eqref{reaction_coord} in Section \ref{preliminaries}), as well as a MCMC kernel targeting $\Prob_{ \mu }( X_{ 0 : \tau_T } | \Psi( X_{ 0 : \tau_T } ) \geq l )$ for any level $l$.
We choose 
\begin{equation*}
\Psi( X_{ 0 : \tau_T } ) \equiv \Psi( ( i_t, j_t )_{ t \in 0 : \tau_T } ) = \max_{ t \in 0 : \tau_T }\{ i_t \}
\end{equation*}
as our reaction coordinate, and specify our MCMC kernel by 
\begin{enumerate}
\item uniformly sampling a time point $t > 0$ along the trajectory,
\item uniformly sampling a new step from $X_{ t - 1 }$ to replace the current value of $X_t$,
\item attaching the remaining trajectory $X_{ ( t + 1 ) : \tau_T }$ to the new value of $X_t$ while keeping its step directions fixed.
\end{enumerate}
Boundary handling was implemented by rejecting any proposals with number of on-sources outside the allowed interval $0 : N$, terminating a proposed trajectory upon its first hitting time of queue length 0 or $b$, and filling in incomplete trajectories from the unconditional ATM queue dynamics if a proposal yielded an incomplete trajectory that had not yet hit $T$.
Proposals whose maximum queue lengths fell below the required level were also rejected.
This adaptive algorithm is known to lead to estimates with near-optimal variance among the class of multilevel splitting algorithms for a given $\Psi$ \citep{Cerou12}, though in practice the variance is also strongly impacted by the MCMC kernel.

\begin{figure}[!ht]
\centering
\includegraphics[width = 0.49\linewidth]{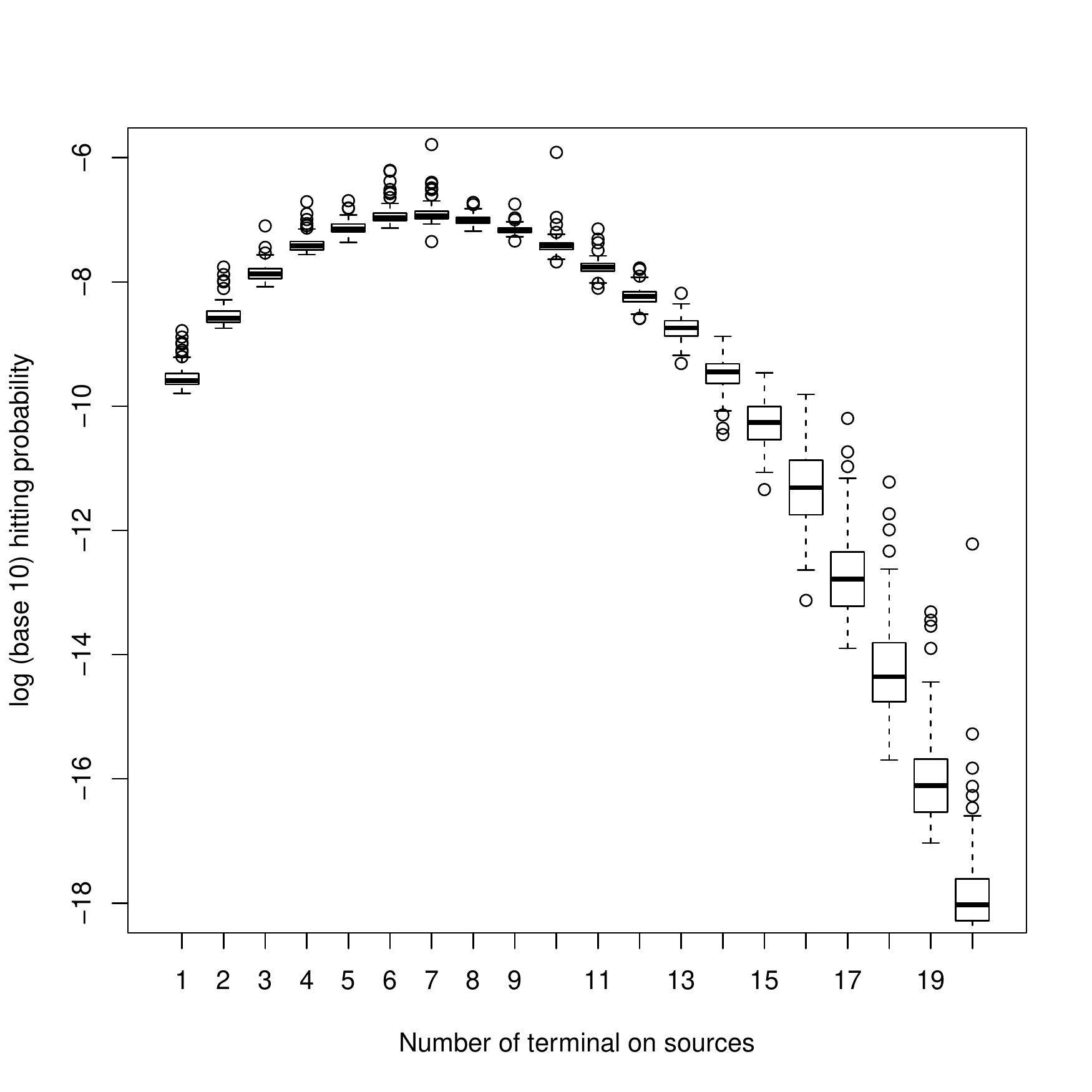}
\includegraphics[width = 0.49\linewidth]{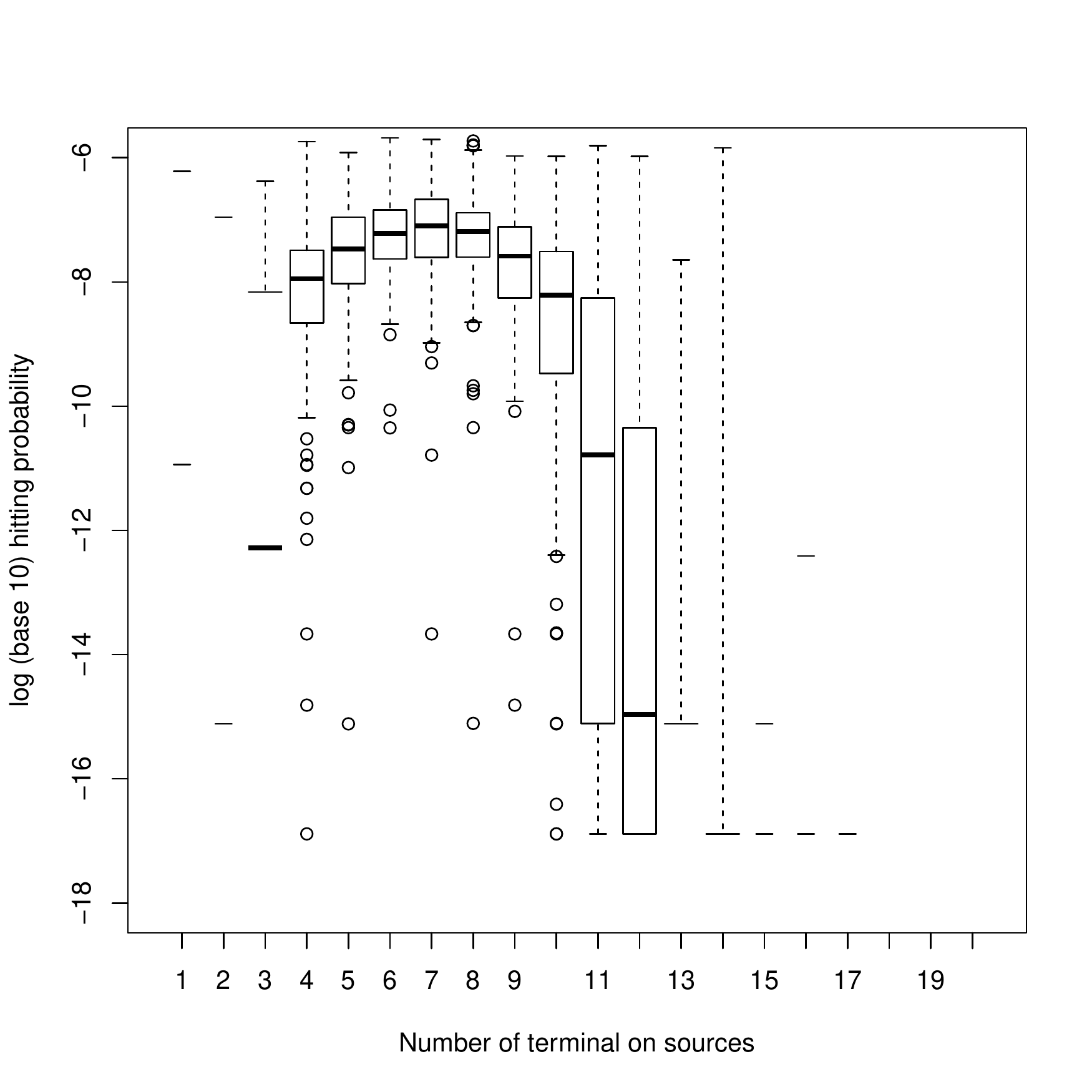}
\caption{100 replicates of simulated hitting probabilities using reverse time SMC (left) and adaptive multilevel splitting (right) on an ATM network with parameters $K = 20$, $b = 10$, $\lambda = 0.5$, $\mu = 10.0$, $\alpha_0 = 1.0$, $\alpha_1 = 3.0$. Particle numbers were 8 000 (SMC) and 10 000 (splitting) respectively, yielding comparable run times of 130 seconds per replicate each on an Intel i5-2520M 2.5 GHz processor. The missing and incomplete boxplots in the bottom figure correspond to terminal values of $j$ that were hit by zero particles  (resulting in a 0 estimate) in too many of the replicates for the default R boxplot method.}
\label{atm_figure}
\end{figure}

Figure \ref{atm_figure} presents 100 replicates of simulated hitting probabilities of a queue length of 10 across all possible fixed numbers of terminal on sources.
The SMC curve shows a high degree of smoothness, strongly suggesting the algorithm has converged, and the Monte Carlo variance is acceptably small for all but the rarest numbers of terminal on-sources.
In contrast, the splitting algorithm has failed to produce estimates for the rarer terminal conditions at all, because too few of the particles being propagated forwards in time managed to hit these regions.
The Monte Carlo variance of the estimates around the mode of the curve are also substantially larger than in the SMC case.

To conclude this section, Figure \ref{atm_fig_2} shows 100 replicates of the reverse time SMC for a more challenging problem.
In this case the adaptive multilevel splitting algorithm could no longer be run in feasible time due to both the length of trajectories and low acceptance probabilities resulting from the rarity of transitions towards the rare event.
Both contributed to slow mixing of the MCMC step, resulting in infeasible run times.
A more efficient MCMC kernel making use of block updates and a more finely tuned proposal mechanism could doubtless improve the performance of the multilevel splitting algorithm, but designing such a proposal is challenging without a detailed understanding of the dynamics of trajectories conditioned on the rare event.
In contrast, reverse time SMC is readily implementable, and remains computationally feasible.

\begin{figure}[!ht]
\centering
\includegraphics[width = 0.49\linewidth]{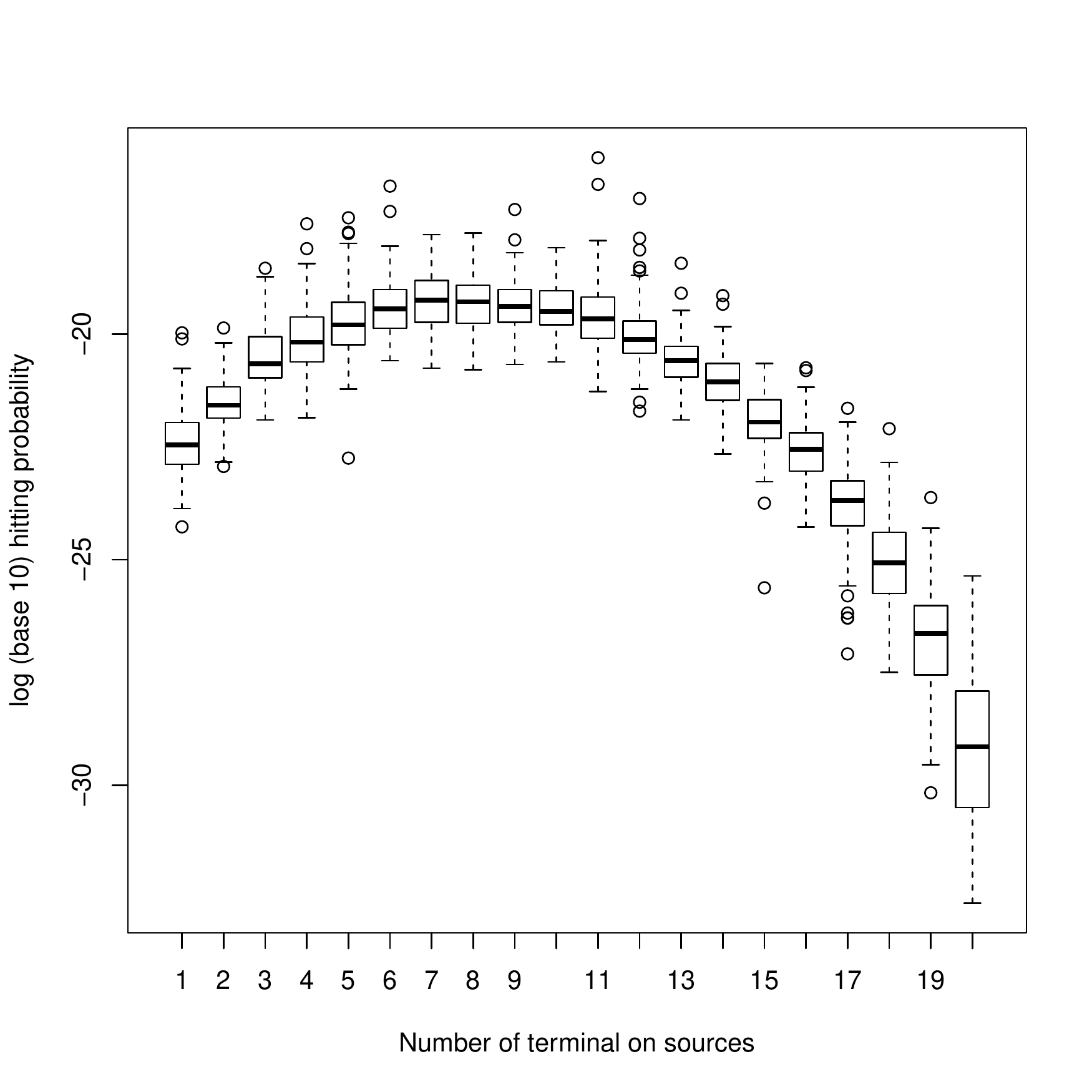}
\caption{100 replicates of simulated hitting probabilities using reverse time SMC on an ATM network with parameters $K = 20$, $b = 30$, $\lambda = 0.5$, $\mu = 10.0$, $\alpha_0 = 1.0$, $\alpha_1 = 3.0$, and 10 000 particles. The total run time was 22 hours, or around 40 seconds per independent simulation (20 terminal conditions $\times$ 100 replicates) on an Intel i5-2520M 2.5 GHz processor.}
\label{atm_fig_2}
\end{figure}

\subsection{The hyperbolic diffusion}\label{hyper_diff}

The scalar hyperbolic diffusion is the solution of the SDE
\begin{equation}\label{hyperbolic_diff}
dX_t = \frac{ - X_t }{ \sqrt{ 1 + X_t^2 } } dt + dW_t,
\end{equation}
where $( W_t )_{ t \geq 0 }$ is a Brownian motion.
It was introduced by \citet{Barndorff78} in connection to hyperbolic distributions in geostatistical modelling \citep{Barndorff77}, and its heavier-than-normal tails have also made it a popular model in mathematical finance \citep{Bibby03}.

The transition probabilities of the diffusion are intractable, but the stationary distribution is known to be the hyperbolic distribution
\begin{equation}\label{hyperbolic_pi}
\pi( x ) = \frac{ 1 }{ 2 K_1( 1 ) } e^{ - \sqrt{ 1 + x^2 } },
\end{equation}
where $K_1$ is the modified Bessel function of the second kind.
We assume that the diffusion is started at stationarity, and focus on the probability that a trajectory lies in an interval $( l_0, u_0 )$ at time 0, and hits interval $( l_t, u_t )$ at time $t \in \N$, without leaving the strip obtained by connecting $l_0$ to $l_t$ and $u_0$ to $u_t$ with straight lines at intermediate times.
Similar containment probabilities have been studied e.g.~by \citet{Casella08} in the context of double barrier option pricing.

Formally, our inference problem is specified by the initial set $I = \{ 0 \} \times ( l_0, u_0 )$, the target set 
\begin{equation*}
T = \bigcup_{ s \in ( 0, t ] } \left( \{ s \} \times \left\{ \frac{ l_t - l_0 }{ t } s + l_0, \frac{ u_t - u_0 }{ t } s + u_0 \right\} \right)
\end{equation*}
the initial distribution $\mu= \pi$, and the quantity of interest 
\begin{equation*}
\E_{ \mu }[ f( X_{ 0 : \tau_T } ) ] = \E_{ \mu }[ \mathds{ 1 }_{ \{ t \} }( \tau_T ) ].
\end{equation*}
We consider a discretisation of \eqref{hyperbolic_diff}, and use the Euler scheme with grid spacing $\Delta $ to define a family of approximate transition densities forwards in time:
\begin{equation}
P_{ \Delta }( ( m, x ), ( n, y ) ) = \frac{ \mathds{ 1 }_{ \{ m + \Delta \} }( n ) }{ \sqrt{ 2 \pi \Delta } } \exp\left( -\frac{ 1}{ 2 \Delta } \left[ y - x \left\{ 1 - \frac{ \Delta }{ \sqrt{ 1 + x^2 } } \right\} \right]^2 \right). \label{euler}
\end{equation}
In this case \eqref{euler} is also a Milstein scheme because of the unit diffusion coefficient.

We can use the discretised transition density \eqref{euler} and the unconditional stationary distribution \eqref{hyperbolic_pi} to define a discretised reverse-time proposal:
\begin{equation*}
\hat{ P }_{ \Delta }( ( n, y ), ( m, x ) ) \propto \frac{ \pi( x ) }{ \pi( y ) } P_{ \Delta }( ( m, x ), ( n, y ) ) \mathds{ 1 }_{ \left( \frac{ l_t - l_0 }{ t } m + l_0, \frac{ u_t - u_0 }{ t } m + u_0 \right) }( x ).
\end{equation*}
We also assume that $\Delta$ divides $t$ exactly, and consider the analogous discretisation of the target set $T$.
We neglect the issue of bias due to unobserved boundary crossings between time discretisation points, though more sophisticated interpolation schemes \cite{Gobet00} could also be implemented.

Note that \eqref{hyperbolic_pi} is not the stationary density of the discrete Euler approximation, which has different stationarity and reversibility properties than the continuous SDE \eqref{hyperbolic_diff}.
However, it can be expected to yield an efficient algorithm, because it is likely that \eqref{hyperbolic_pi} is a good approximation to the occupation measure of the Euler approximation, at least for a sufficiently small time step.
We neglect the bias caused by using \eqref{hyperbolic_pi} as the initial condition of our discretised process.

The proposal distribution can be normalised numerically, and sampled by proposing $x( 1 - \Delta ( 1 + x^2 )^{ -1/2 } )$ from a $\mathcal{ N }( y, \Delta )$ proposal distribution, solving for $x$ and accepting the proposal with probability $e^{1 - \sqrt{ 1 + x^2 } }$.
We have incorporated step 2.~of our strategy (see page 6) automatically in the above definition by rejecting proposed values outside the permitted strip.
We also employ dynamic resampling, in which particles are resampled whenever the effective sample size falls below half the number of trajectories.

As a simulated example, we set $t = 2$, $( l_0, u_0 ) = ( -1, 1 )$, $\Delta = 0.01$ and $( l_t, u_t ) = ( 5, 5.1 )$. 
Note the wide initial condition, the narrow terminal condition and the upward slope of the strip, resulting in rare terminal conditions.
All of these features make time reversal an attractive approach.
Again, we compare our method to the adaptive multilevel splitting algorithm of \citet{Cerou12}, with reaction coordinate 
\begin{equation*}
\Psi( X_{ 0 : \tau_T } ) := \tau_T,
\end{equation*}
and MCMC kernel which, for level $\Psi( X_{ 0 : \tau_T } ) > l$
\begin{enumerate}
\item samples a non-initial time point $s$ uniformly at random,
\item if $s \leq l$, then the proposal is $X_s' \sim U( l_s, u_s )$,
\item else $X_s' \sim P( X_{s - \Delta}, \cdot )$.
\end{enumerate}
Again, proposed trajectories are terminated as soon as they hit the boundary of the target region, and incomplete proposals are propagated using $P$, the unconditional dynamics of the model, until the hitting time $\tau_T$.
Averaging over the initial distribution was done with naive Monte Carlo, with the overall containment probability taken as an average of containment probabilities from uniformly sampled initial conditions weighted by $\mu$.

\begin{figure}[!ht]
\centering
\includegraphics[width = 0.49 \linewidth]{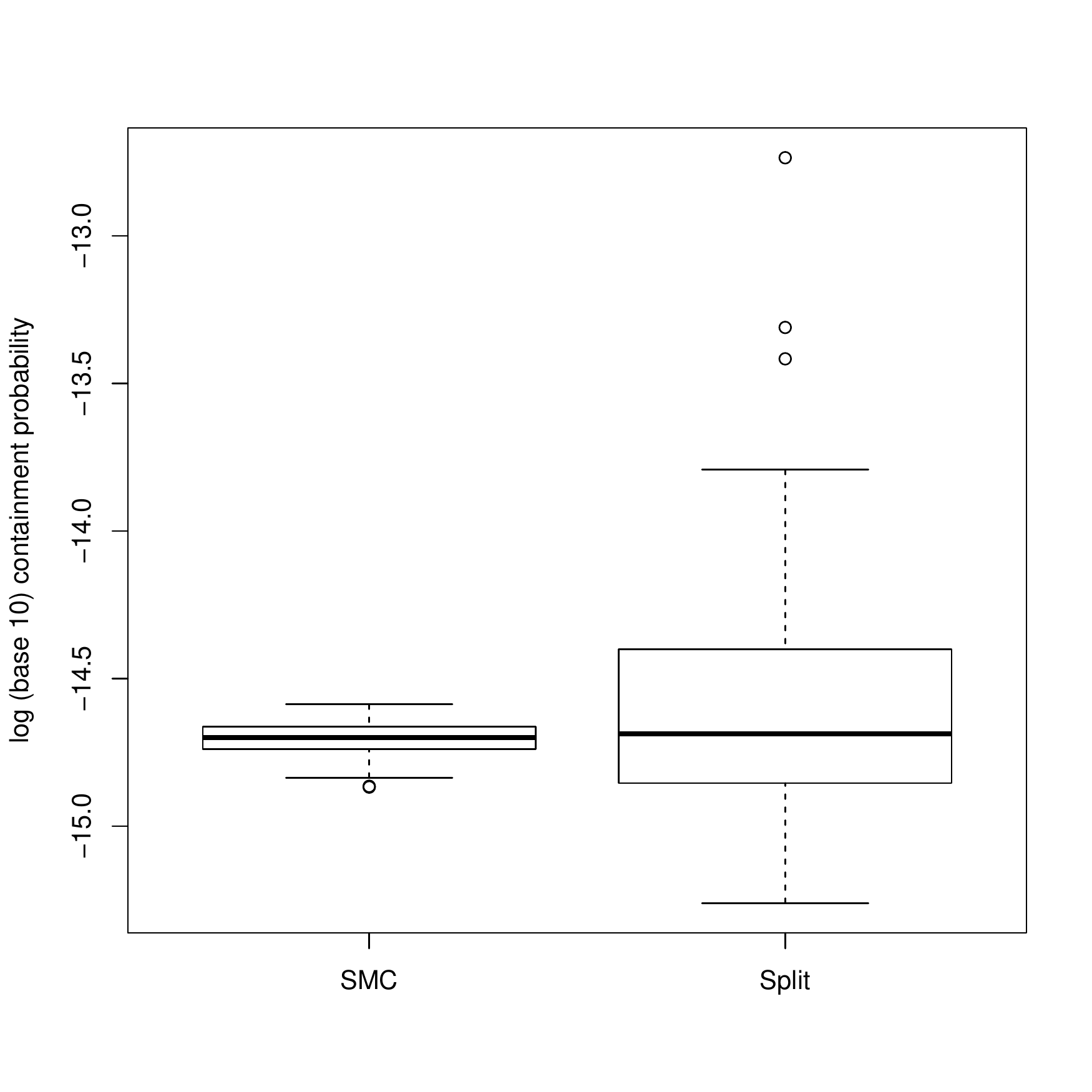}
\caption{100 replicates of simulated containment probabilities of the hyperbolic diffusion with initial interval $( l_0, u_0 ) = ( -1, 1 )$, terminal interval $( l_t, u_t ) = ( 5, 5.1 )$, trajectory length $t = 2$ and time step $\Delta = 0.01$. The particle numbers are 1 000 for SMC and 100 particles for each of 1 000 initial conditons for splitting, giving total run times of 7 seconds (SMC) and 240 seconds (splitting) per replicate on an Intel i5-2520M 2.5 GHz processor.}
\label{diff_figure}
\end{figure}

Figure \ref{diff_figure} shows reverse-time SMC clearly outperforming adaptive multilevel splitting, mainly because the forwards-in-time splitting algorithm has to integrate over a much wider initial condition.
Updating discretised diffusion trajectories using MCMC is also slow because of their length (200 steps in Figure \ref{diff_figure}) combined with the low acceptance probabilities associated with remaining contained in a narrowing interval.
As in the previous section, Figure \ref{diff2} demonstrates the performance of our SMC algorithm in a regime in which our multilevel splitting comparison was no longer feasible due to the increased trajectory length and resulting slower mixing.

\begin{figure}[!ht]
\centering
\includegraphics[width = 0.49 \linewidth]{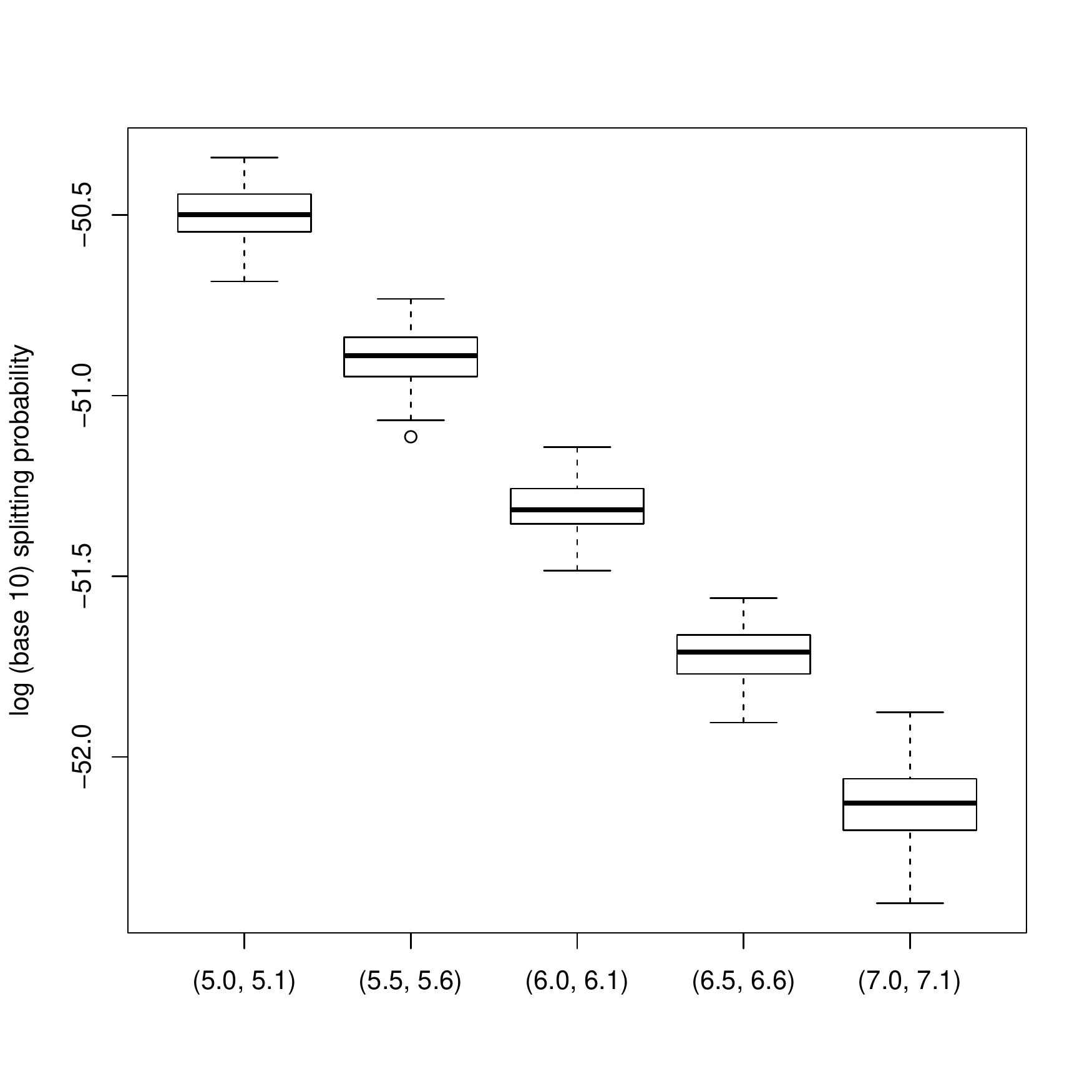}
\caption{100 replicates of simulated containment probabilities of the hyperbolic diffusion with initial interval $( l_0, u_0 ) = ( -1, 1 )$, terminal interval, trajectory length $t = 10$, time step $\Delta = 0.01$ and terminal interval as show. All estimates are using 1 000 particles with per replicate run times of 36, 39, 50, 55, and 66 seconds on an Intel i5-2520M 2.5 GHz processor, respectively.}
\label{diff2}
\end{figure}

\subsection{The susceptible-infected-susceptible network}\label{sis}

Consider a finite network with vertices $V$ and undirected edges $E$, and with vertices labelled as either susceptible ($S$) or infected ($J$).
For a vertex $v \in V$, let $l( v ) \in \{ J, S \}$ denote its label, $N_v := \{ v' \in V : ( v, v' ) \in E \}$ denote its neighbourhood and, for $a \in \{ J, S \}$, let
\begin{equation*}
N_v^{ a } := \{ v' \in V : ( v, v' ) \in E \text{ and } l( v' ) = a \}
\end{equation*}
denote the sub-neighbourhood with label $a$.
Then the susceptible-infected-susceptible (SIS) epidemic evolves as follows.

Every infected node is cured with rate $\beta > 0$, at which point it immediately becomes susceptible again.
A susceptible node is infected by each infected neighbour at rate $\alpha > 0$, so that a vertex $v \in V$ becomes infected at total rate $\alpha | N_v^J |$.
These types of dynamics on networks are popular models e.g.~for the spread of biological epidemics in structured populations \citep{Moore00, Pastor01, Ganesh05}, malware in computer networks \citep{Shah10}, and rumours in social networks \citep{Fuchs15, Shah16}, and are also sometimes referred to as contact processes.
In addition, we impose a rate $\gamma > 0$ at which a new infection enters the network, infecting one uniformly sampled node.
We assume further that new infections can only enter when all vertices are susceptible, i.e.~only one infection can exist in the population at one time.

Suppose that there is no infection in the initial population, and that small infections go undetected.
We define an infection as large once it infects a fixed number $M$ of nodes, and assume that the labels of all nodes are immediately observed as soon as a large infection arises.
Infection times are not observed, nor is any information about the history of the infection, such as whether a vertex that is now susceptible was previously infected.
We are interested in inferring the initial location of the observed large infection, which may no longer be infected itself.
Point estimators for similar inference problems have been studied in \citep{Shah10, Fuchs15, Shah16}.

More formally we consider the jump skeleton of the above continuous time Markov process, let $l_t( v )$ denote the label of vertex $v \in V$ at time $t \in \N_0$, and let the Markov chain $\{ X_t \}_{ t = 0 }^{ \tau_T }$ be given as 
\begin{equation*}
X_t = \{ v \in V : l_t( v ) = J \},
\end{equation*} 
i.e.~the set of infected vertices at time $t$.
Then the initial condition $I$ is the empty set, the target set $T$ is the observed configuration of infected vertices, and the quantity of interest is
\begin{equation*}
\E_{ \mu }[ f( X_{ 0 : \tau_T } ) | \tau_T < \tau_I ] = \E_{ \mu }[ ( \mathds{ 1 }_{ \{ v \} }( X_1 ) )_{ v \in V } | \tau_T < \tau_I ],
\end{equation*}
where $( \mathds{ 1 }_{ \{ v \} }( \cdot ) )_{ v \in V }$ denotes a vector of indicator functions indexed by vertices of the network.
Note that approximating $\E_{ \mu }[ f( X_{ 0 : \tau_T } ) | \tau_T < \tau_I ]$ using a forwards-in-time algorithm is challenging because it can be difficult to know a priori which nodes are likely initial infecteds, and hence the algorithm may spend much effort sampling trajectories of low probability.
The problem also lacks a natural reaction coordinate \eqref{reaction_coord}; the candidate
\begin{equation*}
\Psi( X_{ 0 : \tau_T } ) := \max_{ t \in 0 : \tau_T } \{ | X_t | \}
\end{equation*}
succeeds in creating large epidemics, but does not favour infecting the observed nodes over any others.
Hence, the probability of observing the correct infected configuration is still infeasibly low, even given the correct infection size.
Because nodes can get uninfected and reinfected, the observed infection configuration cannot easily be expressed as the first hitting time of a region of the state space, making this problem very difficult to address.
Neither of these problems causes any difficulty in reverse time: a reaction coordinate is not needed and the initial condition is integrated over automatically by the reverse time algorithm.

It remains to specify the proposal distribution, which we do by specifying a CSD of the label of one vertex given the labels of all the others.
Conditioned on the labels of all other vertices, vertex $v \in V$ becomes infected at fixed rate $\alpha | N_v^J |$ and susceptible with rate $\beta$, so it is natural to base the CSD on the quantities
\begin{equation*}
\left( \frac{ \alpha | N_v^J | }{ \alpha | N_v^J | + \beta }, \frac{ \beta }{ \alpha | N_v^J | + \beta } \right)
\end{equation*}
which are the true stationary probabilities of the label of $v \in V$ being $J$ and $S$ respectively, if the labels of all of its neighbours were fixed.
However, in reverse time a vertex with no infected neighbours can also become infected, corresponding to a forwards-time infection spreading outwards and all connecting individuals becoming uninfected before an isolated leaf infection.
Hence, we add a fixed $\varepsilon > 0$ to the infection rate of every vertex, yielding probabilities
\begin{equation*}
\left( \frac{ \alpha | N_v^J | + \varepsilon }{ \alpha | N_v^J | + \beta + \varepsilon }, \frac{ \beta }{ \alpha | N_v^J | + \beta + \varepsilon } \right).
\end{equation*}
Finally, we also want infections to travel towards a common centre of mass.
To modify our proposal distribution to achieve this,  we require some notation.
For a vertex $v$, let $v_{ \uparrow }$ denote the neighbour directly above, $v_{ \downarrow }$ the neighbour directly below, and analogously let $v_{ \leftarrow }$ and $v_{ \rightarrow }$ deonte the two horizontal neighbours.
Finally, let $d_{\updownarrow}(v) = 1$ if the vertical coordinate of $v$ is lower than the average vertical coordinate of infected vertices in the network, $d_{\updownarrow}(v) = -1$ if the vertical coordinate of $v$ is higher than the average among infected vertices, and $d_{\updownarrow}(v) = 0$ if the two coincide.
Define $d_{\leftrightarrow}(v)$ analogously in the horizontal direction.
Now, let the CSD be given by
\begin{equation*}
\hat{ \pi }( l( v ) | \{ v', l( v' ) \}_{ v' \neq v } ) \propto
\begin{cases}
\frac{ \alpha | N_v^J | + \varepsilon }{ \alpha | N_v^J | + \beta + \varepsilon } w( v ) &\text{ if } l( v ) = J \\
\frac{ \beta }{ \alpha | N_v^J | + \beta + \varepsilon } &\text{ if } l( v ) = S
\end{cases},
\end{equation*}
where the weight $w( v )$ is defined as
\begin{equation*}
w( v ) = 2^{ -[ \mathds{ 1 }_J( l( v_{\uparrow } ) ) - \mathds{ 1 }_J(l( v_{\downarrow} ) ) ] d_{\updownarrow}( v ) - [ \mathds{ 1 }_J( l( v_{\rightarrow } ) ) - \mathds{ 1 }_J( l( v_{\leftarrow} ) ) ] d_{\leftrightarrow}( v ) }.
\end{equation*}
This weight is greater than one when a node that is proposed to be infected lies between an infected neighbour and the centre of mass of the infection, and less than one when a proposed infection is spreading from a neighbour away from the centre of mass.
Hence, it instils a tendency for infections to grow towards a common centre of mass in reverse time.
We note that an approximate CSD based on a larger neighbourhood size would yield a more accurate approximation at greater computational cost.

To study the Monte Carlo variance of our method, Figure \ref{sis_figure} shows the maximum values of estimated likelihood surfaces across 100 iterations and various network sizes.
The locations of the maximisers are summarised in Table \ref{locations}.
In all cases they concentrate on at most two neighbouring vertices with high probability, and invariably both are close to the truth.
We note that there is no reason for the maximum likelihood estimator to coincide with the truth: the quantity of interest is in Table \ref{locations} is the concentration of the maximum likelihood estimators across replicates.
Likewise, there is no reason to expect any particular relationship between maximum likelihood values between network sizes.
The purpose of Figure \ref{sis_figure} is to demonstrate that reverse-time SMC naturally remains tightly concentrated on high probability regions of the path space even as the size of the state space grows.

\begin{figure}[!ht]
\centering
\includegraphics[width = 0.49 \linewidth]{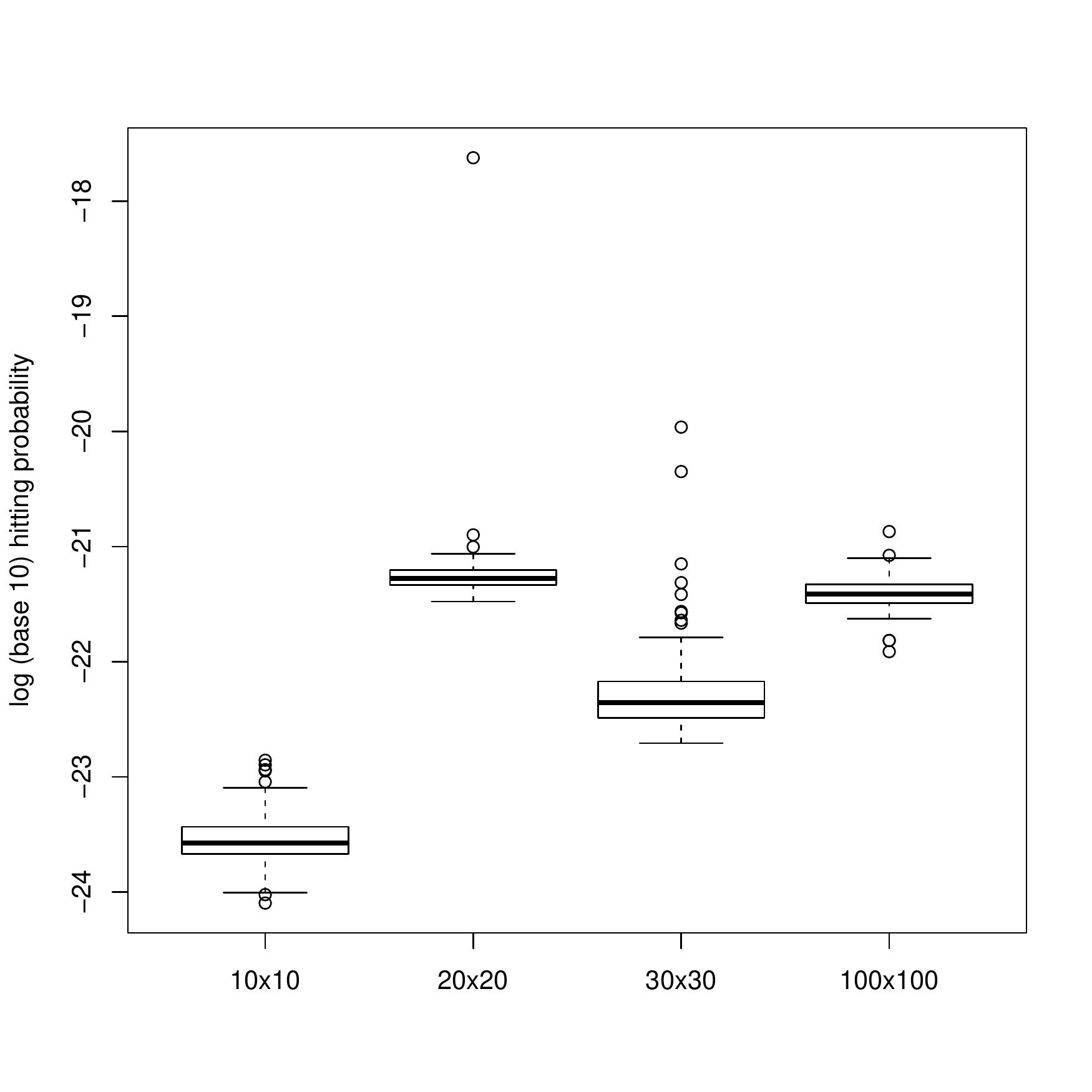}
\caption{Distributions of the maximum likelihood value across 100 iterations as a function of grid size in a nearest neighbour network. The model parameters are $\alpha = 1$, $\beta = 12$, $\gamma = 1$, $\varepsilon = 10^{-4}$ for the $10 \times 10$ and $20 \times 20$ grids, $\varepsilon = 10^{-5}$ for the $30 \times 30$ grid and $\varepsilon = 10^{ -6}$ for the $100 \times 100$ grid, so that $| V | \varepsilon \approx 10^{-2}$. Each simulation used 30 000 particles for per replicate run times of 40, 43, 44, and 89 seconds, respectively, on an Intel i5-2520M 2.5 GHz processor.}
\label{sis_figure}
\end{figure}

\begin{table}[!ht]
\centering
\begin{tabular}{c|c|c | c}
$10 \times 10$ & $20 \times 20$ & $30 \times 30$ & $100 \times 100$ \\
\hline
(7, 3): 57 & (7, 8): 87 & (8, 7): 74 & (77, 50): 98\\
(6, 3): 35 & (8, 7): 9 & (8, 8): 24 & (77, 49) : 1 \\
(5, 3): 6 & (7, 7): 3 & (8, 6): 1 & (77, 51) : 1\\
(8, 3): 1 & (8, 8): 1 & (9, 7): 1 & True: (77, 49)\\
(7, 2): 1 & True: (7, 7) & True: (7, 6)\\
True: (7, 2)

\end{tabular}
\caption{Lists of locations and counts of maximum likelihood estimators of the initial infected position from the experiments described in Figure \ref{sis_figure}. Shown are the row and column labels with an $n \times n$ network represented as $(1 : n) \times (1 : n)$.}
\label{locations}
\end{table}
To conclude this section, Figure \ref{sis_surf} shows a representative likelihood surface for the initial location of the depicted observed epidemic.
This relatively small network already illustrates that the distribution of the initial location is not homogeneous among observed infected locations, and that locations that are no longer infected cannot necessarily be ruled out.
\begin{figure}[!ht]
\centering
\includegraphics[width = 0.49 \linewidth]{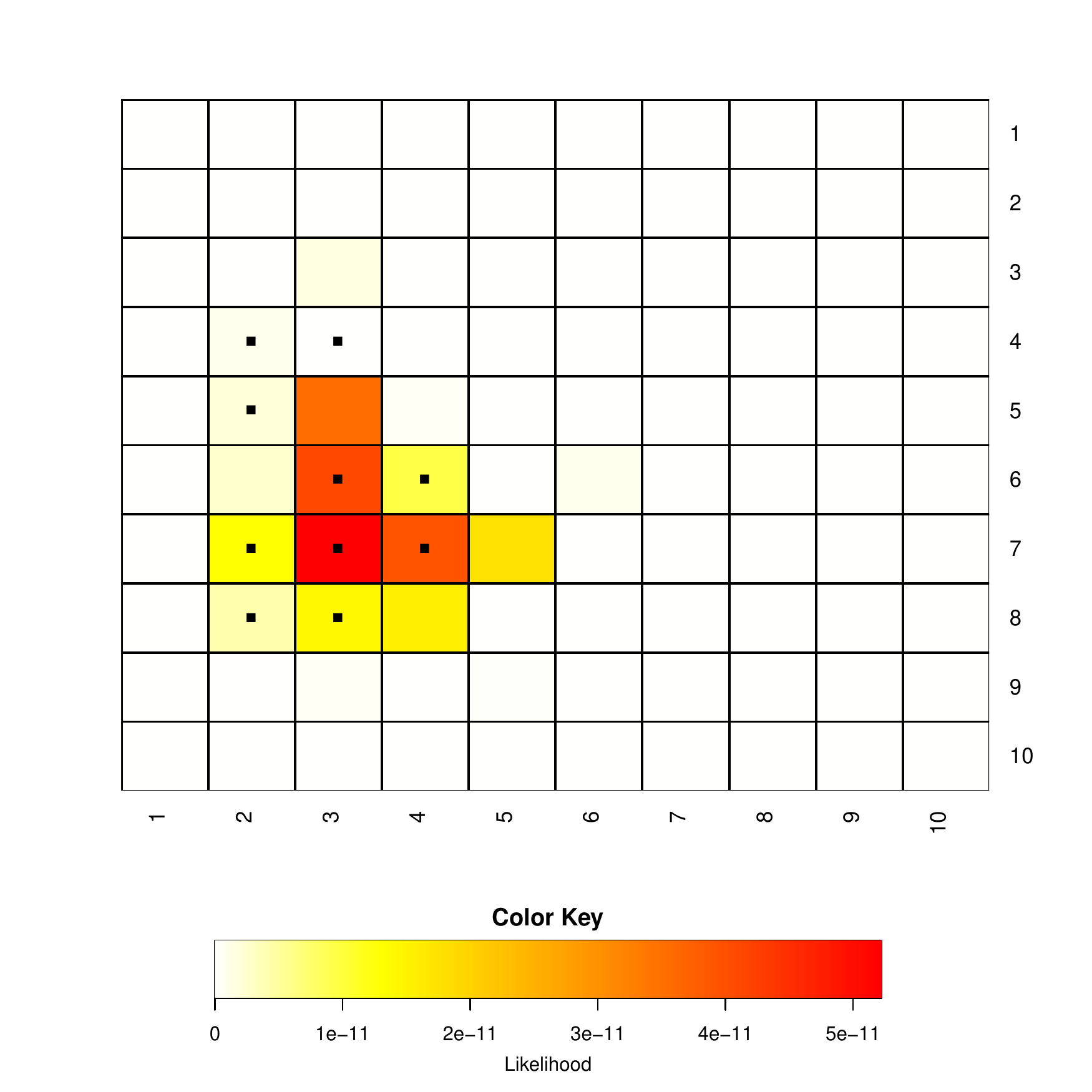}
\caption{A representative simulated likelihood surface for the initial infection position in the 10 $\times$ 10 network. The model parameters are as described in Figure \ref{sis_figure}.}
\label{sis_surf}
\end{figure}

\section{Discussion}\label{discussion}

We have presented a general framework for designing SMC proposal distributions which proceed backwards in time.
Time-reversal makes it straightforward to ensure realisations of the process hit desired regions of the state space, essentially irrespective of the probability assigned to them by the law of the process of interest.
Even the extreme case of conditioning paths on a terminal point of probability 0 can be dealt with easily.
This makes time-reversal a natural and efficient choice when the end point of a path is known with high accuracy, but its initial distribution is flat.
As most existing rare event and path simulation algorithms make the opposite assumptions about initial and terminal conditions, we expect time-reversal to be a useful tool in extending the scope of simulation-based inference and computation.
Our simulation results support this expectation by demonstrating that simple reverse-time SMC schemes outperform a state-of-the-art, adaptive multilevel splitting scheme in settings similar to those described above, and remain feasible for problems for which multilevel splitting requires complex fine tuning.

Expressing the law of the reverse-time process via Nagasawa's formula \eqref{nagasawa} often leads to a substantial reduction in the dimensionality of the design task of defining a proposal distribution (c.f.~Proposition \ref{green_prop}).
The difficulty of designing efficient proposal distributions in high dimension is a central barrier to practical SMC, so this cancellation of dimensions is an important advantage.
Furthermore, we emphasize that it is not inherently linked to time-reversal: re-weighting jump probabilities by an approximate stationary distribution and cancelling out common coordinates would lead to a forwards-in-time proposal defined by low dimensional approximate CSDs.
For rare terminal conditions a reverse-time approach is easier because the conditional and unconditional dynamics share the qualitative behaviour of rapidly leaving the rare state and moving towards a stationary mode.
A forwards-in-time algorithm would have to use CSDs approximating the behaviour of an appropriate Doob's $h$-transform in order to drive the process away from modes and into the rare state.
Nevertheless, we believe analogues of Proposition \ref{green_prop} to be a useful design tool beyond the scope of this paper.

All three example simulations considered in this paper have had the property that the proposal distribution could be normalised numerically, so that proposals could be sampled via standard methods.
This property is computationally convenient, but is often not necessary since importance weights typically only need to be evaluated up to a normalising constant.
Our reverse-time framework inherits this advantage from generic sequential Monte Carlo without added difficulty.
Another advantages which reverse-time SMC inherits from standard SMC are that algorithms are straightforward to parallelise: 
use of GPUs for parallel Monte Carlo simulations has been found to speed up computations by up to 500 fold in comparison to the serial simulations shown in this paper \citep{Lee10}.
Further gains in efficiency can be made by reducing the required number of independent simulations through driving values \citep{Griffiths94c} or bridge sampling \citep{Meng96}.
Such speed up, combined with the facts that
\begin{enumerate}
\item time reversal is advantageous in settings where forwards-in-time methods struggle, and 
\item does not require a reaction coordinate \eqref{reaction_coord} which can be difficult design in practice and is frequently necessary to implement forwards-in-time methods,
\end{enumerate}
leads us to expect that a reverse time perspective can render many previously intractable problems amenable to practical computations.

\section*{Acknowledgements}
The authors are grateful to Adam Johansen and Murray Pollock for fruitful conversations on rare event simulation, and to Ayalvadi Ganesh for pointing out the problem of inferring initial locations on networks.
Jere Koskela was supported by EPSRC as part of the MASDOC DTC at the University of Warwick. Grant No. EP/HO23364/1. Paul Jenkins is supported in part by EPSRC grant EP/L018497/1.

\bibliographystyle{plainnat}
\bibliography{bibliography}  

\end{document}